\RequirePackage{tikz}

\documentclass[pdflatex, sn-mathphys]{sn-jnl}

\jyear{2021}%

\usepackage{amsmath}
\usepackage{amssymb}
\usepackage{bm}
\usepackage{caption}
\usepackage{hyperref}
\usepackage[mathlines]{lineno}
\usepackage{enumitem}
\usepackage{aliascnt}
\usetikzlibrary{automata,positioning}
\usepackage{etoolbox}

\newcommand{\toappendix}[1]{#1}

\theoremstyle{thmstyleone}%
\newtheorem{theorem}{Theorem}

\newaliascnt{lemma}{theorem}
\newtheorem{lemma}[lemma]{Lemma}%
\aliascntresetthe{lemma}

\newaliascnt{corollary}{theorem}
\newtheorem{corollary}[corollary]{Corollary}%
\aliascntresetthe{corollary}

\newaliascnt{observation}{theorem}
\newtheorem{observation}[observation]{Observation}%
\aliascntresetthe{observation}

\theoremstyle{thmstylethree}%
\newtheorem{definition}{Definition}%

\DeclareMathOperator{\ecc}{ecc}
\DeclareMathOperator{\mw}{mw}
\DeclareMathOperator{\treewidth}{tw}
\renewcommand{\|}{\vert}

\newcommand{\problemdef}[3]{
\setlength{\fboxsep}{3pt}
\framebox[\textwidth]{
\begin{minipage}{0.95\textwidth}
\textsc{#1}\\
\textbf{Input:} #2\\
\textbf{Question:} #3
\end{minipage}
}
}

\newcommand*\linenomathpatch[1]{%
  \cspreto{#1}{\linenomath}%
  \cspreto{#1*}{\linenomath}%
  \csappto{end#1}{\endlinenomath}%
  \csappto{end#1*}{\endlinenomath}%
}

\linenomathpatch{equation}
\linenomathpatch{gather}
\linenomathpatch{multline}
\linenomathpatch{align}
\linenomathpatch{alignat}
\linenomathpatch{flalign}

\begin{document}
\title[MESP Problem with Respect to Structural Parameters]{Minimum Eccentricity Shortest Path Problem with Respect to Structural Parameters}

\author*[1]{\fnm{Martin} \sur{Ku\v{c}era} \email{martin@mkucera.cz}}
\author[1]{\fnm{Ond\v{r}ej} \sur{Such\'{y}} \email{ondrej.suchy@fit.cvut.cz}}

\affil[1]{\orgdiv{Department of Theoretical Computer Science}, \orgname{Faculty of Information Technology, Czech Technical University in Prague}, \orgaddress{\street{Thákurova 9}, \city{Prague}, \postcode{160 00}, \country{Czech Republic}}}

\abstract{
The \textsc{Minimum Eccentricity Shortest Path Problem} consists in finding a shortest path with minimum eccentricity in a given undirected graph. The problem is known to be NP-complete and W[2]-hard with respect to the desired eccentricity. We present fpt-algorithms for the problem parameterized by the modular width, distance to cluster graph, the combination of treewidth with the desired eccentricity, and maximum leaf number.
}


\maketitle

\section{Introduction}

The \textsc{Minimum Eccentricity Shortest Path} (MESP) problem asks, given an undirected graph and an integer $k$, to find a shortest path with eccentricity at most $k$---a shortest path (between its endpoints) such that the distance from every vertex in the graph to the nearest vertex on the path is at most~$k$. The shortest path achieving the minimum $k$ may be viewed as the ``most accessible'', and as such, may find applications in communication networks, transportation planning, water resource management, and fluid transportation~\cite{DraganL15}. Some large graphs constructed from reads similarity networks of genomic data appear to have very long shortest paths with low eccentricity~\cite{VolkelBHLV16}. Furthermore, MESP can be used to obtain the best to date approximation for a minimum distortion embedding of a graph into the line~\cite{DraganL15} which has applications in computer vision~\cite{TenenbaumSL00}, computational biology and chemistry~\cite{Indyk01,IndykM04}.
The eccentricity of MESP is closely tied to the notion of laminarity (minimum eccentricity of the graph's diameter)~\cite{BirmeleMP16}.

MESP was introduced by Dragan and Leitert~\cite{DraganL15} who showed that it is NP-hard on general graphs and constructed a slice-wise polynomial~(XP) algorithm, which finds a shortest path with eccentricity at most~$k$ in a graph with $n$ vertices and $m$ edges in $\mathcal{O}(n^{2k+2}m)$ time. They also presented a linear-time algorithm for trees. Additionally, they developed a 2-approximation, a 3-approximation, and an 8-approximation algorithm that runs in $\mathcal{O}(n^3)$ time, $\mathcal{O}(nm)$ time, and $\mathcal{O}(m)$ time, respectively. Birmel\'{e} et al.~\cite{BirmeleMP16} further improved the 8-approximation to a 3-approximation, which still runs in linear time. Dragan and Leitert~\cite{DraganL16} showed that MESP can be solved in linear time for distance-hereditary graphs (generalizing the previous result for trees) and in polynomial time for chordal and dually chordal graphs. Later, they proved~\cite{DraganL17} that the problem is NP-hard even for bipartite subcubic planar graphs, and W[2]-hard with respect to the desired eccentricity for general graphs. Furthermore, they showed that in a graph with a shortest path of eccentricity~$k$, a minimum $k$-dominating set can be found in $n^{\mathcal{O}(k)}$ time. A related problem of finding shortest isometric cycle was studied by Birmel\'{e} et al.~\cite{BirmeleMPV20}. Birmel\'{e} et al.~\cite{BirmeleMPV17} studied a generalization of MESP, where the task is to decompose a graph into subgraphs with bounded shortest-path eccentricity, the hub-laminar decomposition. 

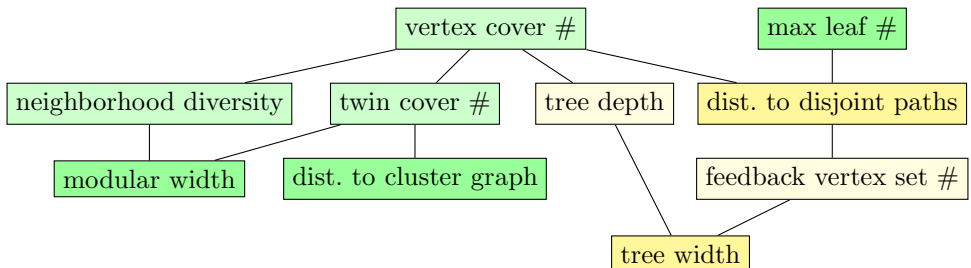
\begin{figure}[t]
\centering
\begin{tikzpicture}[every node/.style={rectangle, draw}]
	\node[fill=green!20] (vc)   at (0,3) {vertex cover \#};
	\node[fill=green!40] (mln)  at (4.5,3) {max leaf \#};
	\node[fill=green!20] (nd)   at (-4.5,2) {neighborhood diversity};
	\node[fill=green!20] (tc)   at (-1,2) {twin cover \#};
	\node[fill=yellow!15] (td)   at (1.5,2) {tree depth};
	\node[fill=green!40] (dtc)  at (-1,1) {dist. to cluster graph};
	\node[fill=yellow!50] (dtdp) at (4.5,2) {dist. to disjoint paths};
	\node[fill=yellow!15] (fvs)  at (4.5,1) {feedback vertex set \#};
	\node[fill=green!40] (mw)   at (-4.5,1) {modular width};
    \node[fill=yellow!50] (tw)   at (2.5,0) {tree width};
	\draw (vc) -- (tc);
	\draw (vc) -- (dtdp);
	\draw (vc) -- (td);
	\draw (tc) -- (dtc);
	\draw (dtdp) -- (fvs);
	\draw (mln) -- (dtdp);
	\draw (nd) -- (mw);
	\draw (vc) -- (nd);
	\draw (tc) -- (mw);
    \draw (tw) -- (td);
    \draw (tw) -- (fvs);
\end{tikzpicture}
\caption{
    Hasse diagram of the boundedness relation between structural parameters explored in this paper and related ones.
    An edge between a parameter $\bm{A}$ above and a parameter $\bm{B}$ below means that whenever $\bm{A}$ is bounded for some graph class, then so is $\bm{B}$.
    The parameters for which MESP is FPT are in green (dark if the result is described in this paper, light if implied by those described).
    Yellow represents a parameter for which MESP is FPT in combination with the desired eccentricity (again, dark if the result is described in this paper, light if implied by those described).
    The figure is inspired by \cite{SorgeW16}.
}
\label{fig:overview}
\end{figure}

\paragraph{Our contribution}
We continue the research direction of MESP in structured graphs~\cite{DraganL16}, focusing on parameters which can measure the amount of structure present in the graph.
We provide fpt-algorithms for the problem with respect to the modular width, distance to cluster graph, distance to disjoint paths combined with the desired eccentricity, treewidth combined with the desired eccentricity, and maximum leaf number (see \autoref{fig:overview} for an overview of our results).

\paragraph{Outline}
In \autoref{chap:preliminaries}, we provide necessary notations and formal definitions. In \autoref{chap:algorithms}, we describe our parameterized algorithms.
In \autoref{chap:conclusion}, we discuss possible future work.


\section{Preliminaries}
\label{chap:preliminaries}
\toappendix{
}
We consider finite connected unweighted undirected simple loopless graphs.

We refer to Diestel~\cite{Diestel2016} for graph notions.

For a~graph $G = (V, E)$ we denote $n = \|V\|$ and $m = \|E\|$.
We denote~$G[S]$ the induced subgraph of $G$ on vertices $S \subseteq V$ and $G \setminus S = G[V \setminus S]$.

We denote an ordered sequence of elements $\bm{s} = (s_1, \dots, s_{\|\bm{s}\|})$. For two sequences $\bm{s} = (s_1, \dots, s_{\|\bm{s}\|}), \bm{t} = (t_1, \dots, t_{\|\bm{t}\|})$ we denote their concatenation
\[
\bm{s} \frown \bm{t} = (s_1, \dots, s_{\|\bm{s}\|}, t_1, \dots, t_{\|\bm{t}\|}).
\]

A \emph{path} is a sequence of vertices where every two consecutive vertices are adjacent.
The first and last vertices of the path are called its \emph{endpoints}. A \emph{path between} $u$ and $v$ or \emph{$u$-$v$-path} is a path with endpoints $u$ and~$v$.
The length of a path $P$ is the number of edges in it, i.e., $\|P\| - 1$.
A $u$-$v$-path is \emph{shortest} if it has the least length among all $u$-$v$-paths.
The distance $d_G(u, v)$ between two vertices $u, v \in V$ is the length of the shortest $u$-$v$-path.

The distance between a vertex $u \in V$ and a set of vertices $S \subseteq V$ is $d_G(u, S) = \min_{s \in S}d_G(u, s)$. 
The eccentricity of a set $S \subseteq V$ is $\ecc_G(S) = \max_{u \in V}d_G(u, S)$.
For a path $P$, we use $P$ instead of $V(P)$ for its set of vertices, if there is no risk of confusion,
e.g., $d_G(u, P) = d_G(u, V(P))$ and $\ecc_G(P) = \ecc_G(V(P))$.

For vertex $u \in V$ we denote $N_G(u)=\{v \mid \{u,v\} \in E\}$ the open neighborhood, $N_G[u]=N_G(u) \cup \{u\}$ the closed neighborhood, and $N_G^k[u] = \{v \in V \mid d_G(u, v) \leq k\}$ the closed $k$-neighborhood of $u$.

In this paper, we focus on the following problem.\\
\problemdef{Minimum Eccentricity Shortest Path Problem (MESP)}
{An undirected graph $G$, desired eccentricity $k \in \mathbb{N}$.}
{Is there a path $P$ in $G$ which is a shortest path between its endpoints with $\ecc_G(P) \leq k$?}

A parameterized problem~$\Pi$ is \emph{fixed parameter tractable~(FPT)} with respect to a parameter~$k$ if there is an algorithm solving any instance of~$\Pi$ with size~$n$ in $f(k) \cdot n^{O(1)}$~time for some computable function~$f$.  Such an algorithm is called a \emph{parameterized} or an \emph{fpt-algorithm}. See Cygan et al.~\cite{CyganFKLMPPS15} for more information on parameterized algorithms.

In this paper, we present fpt-algorithms for MESP with respect to the following structural parameters.

\begin{definition}[Modular width, \cite{GajarskyLO13}]
 Consider graphs that can be obtained from an
algebraic expression that uses the following operations:
\begin{description}
\item[(O1)] create an isolated vertex;
\item[(O2)] the disjoint union of $2$ disjoint graphs (the disjoint union of graphs $G_1$ and $G_2$ is the graph $\big(V(G_1) \cup V(G_2), E(G_1) \cup E(G_2)\big)$);
\item[(O3)] the complete join of $2$ disjoint graphs (the complete
    join of graphs $G_1$ and $G_2$ is the graph $\big(V(G_1) \cup V(G_2), E(G_1) \cup E(G_2) \cup \big\{ \{v,w\} \mid v \in V(G_1), w  \in V(G_2) \big\}\big)$);
\item[(O4)] the substitution with respect to some pattern graph $T$
(for a graph $T$ with
  vertices $t_1,\dots,t_n$ and disjoint graphs $G_1,\dots,G_n$ the
  substitution of the vertices of $T$ by the graphs
  $G_1,\dots,G_n$ is the graph with
  vertex set $\bigcup_{i=1}^n V(G_i)$ and edge set $\bigcup_{i=1}^n E(G_i) \cup \big\{ \{u,v\} \mid u \in V(G_i), v \in V(G_j) \textup{, and }\{t_i, t_j\} \in E(T) \big\}$).
\end{description}
We define the \emph{width} of an algebraic expression $A$ as the maximum number of
operands used by any occurrence of the operation (O4) in $A$.
The \emph{modular-width} of a graph $G$, denoted
$\mw(G)$, can be defined as the least integer $m$ such that $G$ can be obtained from
such an algebraic expression of width at most $m$. 
\end{definition}
Given a graph $G$ with $n$ vertices and $m$ edges, an algebraic expression of width $\mw(G)$ describing~$G$ can be constructed in $\mathcal{O}(n+m)$ time~\cite{TedderCHP08}.

\begin{definition}[Distance to cluster graph]
For a graph $G = (V, E)$, a \emph{modulator to cluster graph} is a vertex subset $X \subseteq V$ such that $G \setminus X$ is a vertex-disjoint union of cliques. The \emph{distance to cluster graph} is the size of the smallest modulator to cluster graph.
\end{definition}
A modulator to cluster graph of a graph with distance to cluster graph $p$ can be found in $\mathcal{O}\big(1.9102^p \cdot (n+m)\big)$ time.~\cite{BoralCKP16}

\begin{definition}[Distance to disjoint paths]
For a graph $G = (V, E)$ a \emph{modulator to disjoint paths} is a  vertex subset $X \subseteq V$, such that $G \setminus X$ is a vertex-disjoint union of paths. The \emph{distance to disjoint paths} is the size of the smallest modulator to disjoint paths.
\end{definition}

For completeness, we include the following result which is rather folklore.

\begin{lemma}
\label{alg:modulator-to-disjoint-paths}
The modulator to disjoint paths $C$ of a graph $G$ with distance to disjoint paths $c$ can be found in $\mathcal{O}\big(4^c(n+m)\big)$ time.
\end{lemma}

\toappendix{

\begin{proof}
If the highest degree in $G$ is at most 2, then $G$ consists only of disjoint paths and cycles, and the modulator to disjoint paths is a set of vertices containing one vertex from each cycle. Thus, the modulator to disjoint paths can be found in $\mathcal{O}(n+m)$ time by identifying all cycles with a depth-first search.

If the highest degree in $G$ is at least 3, then the modulator to disjoint paths can be found by a simple branching rule.
\begin{enumerate}
    \item Select any vertex $u$ with degree $\deg_G(u) \geq 3.$
    \item Either $u \in C$ or some subset $S \subseteq N_G(u)$ of size $\|S\| = \deg_G(u) - 2$ must be in $C$.
\end{enumerate}

We show that this algorithm has time complexity $4^cq(n+m)$ for some constant $q$, where $c$ is the given maximum distance to disjoint paths. We show that by induction on $c$. For $c=0$ we only have to check whether the graph is a disjoint union of paths, which can be done in $\mathcal{O}(n + m)$ time, i.e., $4^0q(n+m)$ time for a suitably chosen $q$.

For $c\ge 1$ we have $T(c)    \le    T(c - 1) + \binom{d}{d-2} T(c-d+2)$ for some $d\ge 3$, which, by induction hypothesis is at most $(4^{c-1}+\binom{d}{d-2}4^{c-d+2})q(n+m)$. To show that this is at most $4^cq(n+m)$, it remains to show that $\binom{d}{d-2}4^{-d+2} \le 1-4^{-1}$, i.e., that $\binom{d}{2} \le 3 \cdot 4^{d-3}$ for every $d \ge 3$. For $d \ge 5$ we have $\binom{d}{2} \le 2^d \le 2\cdot 2^{2d-6} \le 3 \cdot 4^{d-3}$, whereas for $d=3$ we have $\binom{d}{2} = 3 = 3 \cdot 4^{d-3}$ and for $d=4$ we have $\binom{d}{2} = 6 \le 12= 3 \cdot 4^{d-3}$.

Hence, the time complexity is indeed $\mathcal{O}(4^c (n+m))$.
\end{proof}
}

\begin{definition}[Treewidth]
A \emph{tree decomposition} of a graph $G = (V, E)$ is a tuple $(T, \beta)$ where $T$ is a tree and $\beta: V(T) \to 2^V$ such that
\begin{enumerate}
    \item $\bigcup_{x \in V(T)} \beta(x) = V$,
    \item $\forall \{u, v\} \in E\ \exists x \in V(T): \{u, v\} \subseteq \beta(x)$, and
    \item $\forall v \in V: $ nodes $\{x \in V(T) \mid v \in \beta(x)\}$ induce a connected subtree of $T$.
\end{enumerate}
The \emph{width} of a tree decomposition $(T, \beta)$ is $\max\{\|\beta(x)\|-1 : x \in V(T)\}$.
The \emph{treewidth} $\treewidth(G)$ of a graph $G$ is the smallest width over all tree decomposition of $G$.
\end{definition}

\begin{definition}[Maximum leaf number]
    The \emph{maximum leaf number} of a graph~G is the maximum number of leaves in a spanning tree of G.
\end{definition}

The presented algorithms rely on the following lemma.

\begin{lemma}
\label{th:unique-permutation}
For any graph $G = (V, E)$, any set $S \subseteq V$, and any vertex $s \in V$, at most one permutation $\bm\pi = (\pi_1, \dots, \pi_{\|S\|})$ of the vertices in $S$ exists, such that there is a shortest path $P$ with the following properties:
\begin{enumerate}
    \item The first vertex on $P$ is $s$,
    \item $P$ contains all vertices from $S$, and
    \item the vertices from $S$ appear on $P$ in exactly the order given by $\bm\pi$.
\end{enumerate}
Moreover, given a precomputed distance matrix for $G$, the permutation $\bm\pi$ can be found in $\mathcal{O}(\|S\| \log \|S\|)$ time.
\end{lemma}

\begin{proof}
For the sake of deriving a contradiction, suppose that there are two different permutations $\bm\pi$ and $\bm\pi'$ satisfying the conditions, and let $P, P'$ be the respective shortest paths.
Let $i \in \{1, \dots, \|S\| - 1\}$ be the first position such that $\pi_i \ne \pi'_i$ and let $j \in \{2, \dots, \|S\|\}$ be the position of $\pi_i$ in $\bm\pi'$ (clearly, $j > i$).
Let $P_1$ be the subpath of $P$ from $s$ to $\pi_i$, and $P'_2$ the subpath of $P'$ from $\pi'_j$ to $\pi'_{\|S\|}$ (excluding the first vertex $\pi'_j$).
Then, $P'' = P_1 \frown P'_2$ is a path which is strictly shorter than $P'$ and has the same endpoints.
This contradicts $P'$ being a shortest~path.

Sorting all vertices in $S$ by increasing distance from the starting endpoint $s$ yields our permutation $\bm\pi$.
It corresponds to some shortest path $P$ if and only if
\[d_G(s, \pi_1) + \textstyle\sum_{i=1}^{\|S\| - 1}d_G(\pi_i, \pi_{i+1}) = d_G(s, \pi_{\|S\|}).\]
\end{proof}

\section{Parameterized Algorithms}
\label{chap:algorithms}
In this section, we present several fpt-algorithms for MESP.
In \autoref{sec:modular-width}, we present an algorithm parameterized by the modular width.
In \autoref{sec:csc} we define the \textsc{Constrained Set Cover} (CSC) problem.
In \autoref{sec:distance-to-cluster} we show an fpt-algorithm for MESP parameterized by the distance to cluster graph which reduces MESP to CSC.
In \autoref{sec:distance-to-disjoint-paths}, we present an fpt-algorithm parameterized by the distance to disjoint paths and the desired eccentricity, combined.
This algorithm also depends on the solution of the CSC problem.
In \autoref{sec:treewidth}, we present an fpt-algorithm parametrized by treewidth and the desired eccentricity, combined.
In \autoref{sec:max-leaf-number} we present an algorithm parameterized by the maximum leaf~number.

\subsection{Modular Width}
\label{sec:modular-width}
\toappendix{
}
We present an fpt-algorithm for MESP parameterized by the modular width.

Let $G = (V, E)$ be a graph with modular width $w$ and $A$ be the corresponding algebraic expression describing the graph. We take a look at the last operation applied in $A$. Operation (O1) is trivial and (O2) yields a disconnected graph, therefore we suppose the last operation is either (O3) or (O4).

If it is (O3) and $G$ is a path (of length at most 3), then the whole path is trivially a shortest path with eccentricity 0.
If $G$ is not a path, then the minimum eccentricity shortest path is any single edge connecting the two original graphs with eccentricity 1.

If it is (O4), the pattern graph is $T = (V_T, E_T)$ with $V_T = \{v_1, \dots, v_w\}$ and the substituted graphs are $G_1, \dots, G_w$, then we suppose that $w \geq 3$ and $T$ is not a clique (otherwise (O3) could be used as the last operation). We continue by showing that the structure of the pattern graph restricts the structure of any shortest path in the resulting graph significantly.

\begin{lemma}
\label{th:single_vertex} 
If the last operation in $A$ is (O4), then there is a minimum eccentricity shortest path in $G$ which contains at most one vertex from each $G_i$ for $i \in \{1, \dots, w\}$.
\end{lemma}

\toappendix{
\begin{proof}
Let $P$ be a shortest path in $G$.

If the length of $P$ is at most 2 and $P$ contains two vertices from $G_i$, then we create a path $P'$ with $\ecc_G(P') \leq \ecc_G(P)$. Because $G$ was created with (O4), we know that $G \not\simeq P$, thus $\ecc_G(P) \geq 1$. We denote $a, b$ the two vertices from $G_i$ on $P$. If there is another vertex on $x$ on $P$ and $x \notin G_i$, we let $c = x$; otherwise we choose a vertex $c$ arbitrarily from some $G_j$ such that $j \neq i$ and $\{v_i, v_j\} \in E_T$. Then, we let $P' = (b, c)$. For every vertex $u \in G_i$ we have $d_G(u, P') \leq 1$ and for every vertex $v \notin G_i$ we have $d_G(v, P') \leq d_G(v, P)$.

If the length of $P$ is at least 3, we show by contradiction that it cannot contain two vertices from $G_i$. Let $P = (p_1, \dots, p_\ell)$. Suppose that $p_s, p_t \in G_i$. Clearly, at least one of $p_s, p_t$ is not an endpoint of $P$ (otherwise the length of $P$ would be at most~2). Without loss of generality, suppose that $p_s$ is not an endpoint of $P$, thus it has a predecessor $p_{s-1}$ and $P = (p_1, \dots, p_{s-1}, p_s, \dots, p_t, \dots, p_\ell)$. We have $\{p_{s-1}, p_s\} \in E$ and $\{p_{s-1}, p_t\} \in E$. Path $P$ may be shortened to $P' = (p_1, \dots, p_{s-1}, p_t, \dots, p_\ell)$. Thus, $P$ is not a shortest path.
\end{proof}
}

Now, we show that with respect to eccentricity, all vertices in the same graph~$G_i$ are equivalent. That means a minimum eccentricity shortest path in~$G$ can be found by trying all shortest paths in $T$.

\begin{lemma}
\label{th:modular-width--substitute-p}
Let $P$ be a shortest path in $G$ and $p \in P \cap G_i$. We create a path $P'$ by substituting $p$ in $P$ by any $p' \in G_i$. Then, $\ecc_G(P') = \ecc_G(P)$.
\end{lemma}

\toappendix{
\begin{proof}
Let $u \in V(G)$ such that $u \neq p$ and $u \neq p'$.
If $u \notin G_i$, then $d_G(u, p) = d_G(u, p')$ and thus $d_G(u, P) = d_G(u, P')$.
If $u \in G_i$, then $d_G(u, P) \leq 1$ and $d_G(u, P') \leq 1$ because the neighbors of $p$ on $P$, as well as the neighbors of $p'$ on $P'$, are also neighbors of $u$.
Moreover, $d_G(p, P') = d_G(p', P) = 1$.
\end{proof}
}

Based on what we have shown, we can construct an algorithm to solve MESP. We handle separately the graphs created using (O1) or (O3) as the last operation.
For (O4), we iterate through all possible shortest paths $\bm\pi$ in $T$.
For each of them and each $i \in \{1, \dots, \|\bm\pi\|\}$ we let $p_i \in G_{\pi_i}$ arbitrarily, and let $P := (p_1, \dots, p_{\|\bm\pi\|})$.
Then we check whether $P$ is a shortest path with eccentricity at most $k$ in~$G$. By the above arguments, if there is a shortest path of eccentricity at most $k$, we will find one.

All shortest paths in a graph can be found by simply performing $n$ DFS traversals (one starting in each vertex).
Each forward step of the DFS represents a new shortest path; we skip edges that would break the shortest path property (this can easily be checked with a precomputed distance matrix).

By assuming a trivial upper bound $2^w$ on the number of shortest paths in $T$, we arrive at the following theorem.

\begin{theorem}
There is an algorithm that solves MESP in $\mathcal{O}(2^w \cdot n^3)$ time, where $w$ is the modular width of the input graph.
\end{theorem}

\subsection{Constrained Set Cover}
\label{sec:csc}
\toappendix{
}
In this subsection we define the \textsc{Constrained Set Cover} (CSC) problem. In the folllowing two subsections it will be used as a subroutine to solve MESP.\\[\medskipamount]
\problemdef{Constrained Set Cover}{A set $\mathcal{C} = C_1 \cup \dots \cup C_m$ of candidates, a set $\mathcal{R} = \{r_1, \dots, r_n\}$ of requirements to be satisfied, and a function $\Psi: \mathcal{C} \rightarrow 2^\mathcal{R}$ that determines for each candidate which requirements it satisfies.}
{Is there a \emph{constrained set cover}, that is, a set of candidates, exactly one from each set
$s_1 \in C_1, \dots, s_m \in C_m$ such that together they satisfy all the requirements, i.e., $\Psi(s_1) \cup \dots \cup \Psi(s_m) = \mathcal{R}$?}

Each candidate can be thought of as a set of (satisfied) requirements. 
Hence, if we drop the constraints $s_i \in C_i$, we get the ordinary \textsc{Set Cover} with a universe $\mathcal{R}$, a family of sets $\{\Psi(c) \mid c \in \mathcal{C}\}$, and the question of whether there is a set cover of size at most $m$.
In our definition several candidates can satisfy the same set of requirements.

In the next two subsections we use the following theorem.

\begin{theorem}
\label{lem:csc}
\textsc{Constrained Set Cover} can be solved in $\mathcal{O}(2^{2\|\mathcal{R}\|} \|\mathcal{R}\| \cdot \|\mathcal{C}\|)$ time.
\end{theorem}

\toappendix{
The rest of this subsection is devoted to the proof of \autoref{lem:csc}.
To prove \autoref{lem:csc}, we need some further definitions and lemmas.

To help us solve CSC, we now define a function $D_i$ for each $i \in \{1, \dots, m\}$. 

\begin{definition}
\label{def:csc-dyn}
Let $\mathcal{R} = \{r_1, \dots, r_n\}, \mathcal{C} = C_1 \cup \dots \cup C_m, \Psi: \mathcal{C} \rightarrow 2^\mathcal{R}$ be an instance of CSC. We define function $D_i: 2^\mathcal{R} \rightarrow \mathcal{C} \cup \{\top, \bot\}$ as follows:
Given some requirements~$R\subseteq {\mathcal R}$,
\[D_i(R) = \begin{cases} 
    \top    & \text{if } i = 0 \land R = \emptyset\\
    s_i     & \text{if } \exists\ (s_1 \in C_1, \dots, s_i \in C_i): R \subseteq \Psi(s_1) \cup \dots \cup \Psi(s_i)\\
    \bot    & \text{otherwise}.
\end{cases}\]
\end{definition}

If there are more candidates in the second case, we select an arbitrary one to make $D_i$ a function.
As we will see, it does not matter which specific value it has, as long as it satisfies the definition.

Before using this function to solve CSC, we need to know how to calculate its values efficiently.
Note that $D_0(R)=\top$ if $R=\emptyset$ and $\bot$ otherwise.

\begin{lemma}
\label{th:csc--recursion}
Let $i \ge 0$. Then $D_{i+1}$ can be computed recursively as follows.
\[D_{i+1}(R) = \begin{cases}
    s_{i+1} & \text{if } \exists\ K \subseteq R, s_{i+1} \in C_{i+1}: R \subseteq K \cup \Psi(s_{i+1}) \land D_i(K) \neq \bot\\
    \bot & \text{otherwise}
\end{cases}\]
\end{lemma}

\begin{proof}
We suppose that the initial values of $D_0$ are taken directly from \autoref{def:csc-dyn} and prove the lemma for any $i \geq 0$. First, we show that if the recursion yields some value $s_{i+1}$, then it is one of the correct possible values of $D_{i+1}(R)$ according to \autoref{def:csc-dyn}. Second, we show that if $D_i(R) \neq \bot$, then the recursion does not yield $\bot$ either.

If the recursion yields $s_{i+1} \in C_{i+1}$, then there is some $K \subseteq R$ such that $R \subseteq K \cup \Psi(s_{i+1}) \land D_i(K) \neq \bot$.
From induction, the value of $D_i(K)$ is correct, thus there exist $(s_1 \in C_1, \dots, s_i \in C_i): K \subseteq \Psi(s_1) \cup \dots \cup \Psi(s_i)$. Hence, $R \subseteq \Psi(s_1) \cup \dots \cup \Psi(s_i) \cup \Psi(s_{i+1})$, which corresponds to \autoref{def:csc-dyn} and $D_{i+1}(R)$ yields a correct value.

If $D_{i+1}(R) = s_{i+1} \in C_{i+1}$, then there are some candidates $s_1 \in C_1, \dots s_{i+1} \in C_{i+1}$ such that $R \subseteq \Psi(s_1) \cup \dots \cup \Psi(s_{i+1})$. Let $K = \Psi(s_1) \cup \dots \cup \Psi(s_i)$. Then, $R \subseteq K \cup \Psi({s_i+1})$ and $D_i(K) = s_i$. Hence the recursion does not yield $\bot$.
\end{proof}

We continue by showing how a solution of CSC may be extracted from the values of $D_1, \dots, D_m$. We will use each function $D_i$ to choose the candidate $s_i$ from $C_i$.

\begin{lemma}
\label{th:constrained-set-cover-dp}
If $D_m(\mathcal{R}) = \bot$, then no solution exists. Otherwise, the solution can be found by iterating through the calculated values backwards and setting:
$s_m = D_m(\mathcal{R})$, 
$s_{m-1}  = D_{m-1}(\mathcal{R} \setminus \Psi(s_m))$,
            \dots,
           $ s_i     = D_i(\mathcal{R} \setminus \bigcup_{j=i+1}^m \Psi(s_j) )$,
             \dots,
             $s_1      = D_1(\mathcal{R} \setminus \bigcup_{j=2}^m \Psi(s_j) )$.
\end{lemma}

\begin{proof}
By the definition of $D_m$, we only have $D_m(\mathcal{R}) = \bot$ if no set of candidates $s_1 \in C_1, \dots, s_m \in C_m$ exists, such that $\mathcal{R} \subseteq \Psi(s_1) \cup \dots \cup \Psi(s_m)$, i.e., if no solution of the CSC instance exists.

Otherwise, using the definition of $D_m$, we can set $s_m = D_m(\mathcal{R})$.
We know that given an $i \in \{1, \dots, m - 1\}$, all the requirements $\bigcup_{j=i+1}^m \Psi(s_j)$ are satisfied by $s_{i+1}, \dots, s_m$, so the rest of them needs to be satisfied by $s_1, \dots, s_i$. Also, by the definition of $D_{i+1}$ we know that there exists $s_1, \dots, s_i$ such that they satisfy the rest of the requirements. Then, by the definition of $D_i$, we have $D_i\big(\mathcal{R} \setminus \bigcup_{j=i+1}^m \Psi(s_j) \big) = s_i$.
\end{proof}

Finally, we propose \autoref{alg:constrained-set-cover} to solve CSC using dynamic programming. We first compute all values for each of $D_1, \dots, D_m$ using the recursion from \autoref{th:csc--recursion}, and then construct the solution $s_1, \dots, s_m$ from them as in \autoref{th:constrained-set-cover-dp}.
\autoref{lem:csc} summarizes the properties of \autoref{alg:constrained-set-cover}.

\begin{algorithm}[t!]
	\caption{Constrained Set Cover}
	\label{alg:constrained-set-cover}
    \begin{algorithmic}[1]
	\Require Set of requirements $\mathcal{R}$, sets of candidates $\mathcal{C} = C_1 \cup \dots \cup C_m$, function $\Psi: \mathcal{C} \rightarrow 2^\mathcal{R}$
	\ForAll{$S \subseteq \mathcal{R}$}
	    \State{$D_0(S) \leftarrow \bot$}
	\EndFor
    \State{$D_0(\emptyset) \leftarrow \top$}
    \For{$i = 1, \dots, m$}
        \ForAll{$S \subseteq \mathcal{R}$}
            \State{$D_i(S) \leftarrow \bot$}
        \EndFor
        \ForAll{$c \in C_i$}
            \ForAll{$K \subseteq \mathcal{R}, F \subseteq \Psi(c)$}
                \If{$D_{i-1}(K) \neq \bot$}
                    \State{$D_i(K \cup F) \leftarrow c$}
                \EndIf
            \EndFor
        \EndFor
    \EndFor
    \If{$D_m(\mathcal{R}) \neq \bot$}
        \For{$i = m, \dots, 1$}
            \State{$s_i \leftarrow D_i(\mathcal{R})$}
            \State{$\mathcal{R} \leftarrow \mathcal{R} \setminus \Psi(s_i)$}
        \EndFor
        \State{\Return $(s_1, \dots, s_m)$}
    \Else
        \State{No solution exists.}
    \EndIf
\end{algorithmic}
\end{algorithm}

\begin{proof}[Proof of \autoref{lem:csc}]
We prove that \autoref{alg:constrained-set-cover} solves \textsc{Constrained Set Cover} in $\mathcal{O}(2^{2\|\mathcal{R}\|} \|\mathcal{R}\| \cdot \|\mathcal{C}\|)$ time.

On lines 1--3 we set $D_0$ according to \autoref{def:csc-dyn}. On lines 4--10 we compute $D_1, \dots, D_m$ according to \autoref{th:csc--recursion}. On lines 11--17 we construct the solution from $D_1, \dots, D_m$ according to \autoref{th:constrained-set-cover-dp}.

The for-loop on line 4 iterates over all sets of candidates and the foreach-loop on line 7 iterates over all candidates in each set. Together, lines 8--10 will be executed $\|\mathcal{C}\|$ times. The foreach-loop on line 8 iterates over all subsets of $\mathcal{R}$ and all subsets of the output of $\Psi$, which sums to at most $\mathcal{O}(2^{\|\mathcal{R}\|} \cdot 2^{\|\mathcal{R}\|})$ iterations in total. Lines 9--10 can be implemented in $\mathcal{O}(\|\mathcal{R}\|)$ time. The for-loop on line 12 has $m$ iterations, and lines 13--14 can be implemented in $\mathcal{O}(\|\mathcal{R}\|)$ time.
\end{proof}
}

\subsection{Distance to Cluster Graph}
\label{sec:distance-to-cluster}
\toappendix{
}
In this subsection, we present an fpt-algorithm for MESP parameterized by the distance to cluster graph.
The trivial case where distance to cluster graph is 0 is omitted.
%
Note that if $G$ is a graph with a modulator to cluster graph~$U$, then, for any edge $\{u, v\}$ in $G \setminus U$, $u$ and $v$ have the same neighborhood in $G \setminus U$.

The high-level idea of the algorithm is that we iteratively guess (by trying all possible combinations), for each vertex in the modulator to cluster $U$, whether it lies on the desired shortest path (we say it belongs to the set $L$), or it is at distance 1 or 2 from the shortest path (it belongs to the set $\mathcal{R}_1$ or $\mathcal{R}_2$, respectively), or at an even further distance.
Then, we try to find a shortest path such that all vertices from $L$ lie on it, and all vertices in $\mathcal{R}_1, \mathcal{R}_2$ have the respective distance from the path.
Finding such a path is reduced to solving the CSC problem presented in \autoref{sec:csc}.
Once we guess the correct combination of these sets, we actually construct the MESP.

First, we discuss some properties of graphs having the desired path. 

\begin{lemma}
\label{th:distance-to-cluster--contains-U}
Let $G$ be a graph with modulator to cluster graph $U$ and let $P$ be a shortest path with $\ecc_G(P) = k$. Then, there exists a shortest path $P'$ such that it contains at least one vertex from $U$ and $\ecc_G(P') \leq k$.
\end{lemma}

\begin{proof}
Suppose that $P$ only contains vertices from $V= V(G) \setminus U$. All these vertices form a clique, so the length of $P$ is at most 1. Let $P = (u, v)$. If there is some vertex $w \in U$ such that it is a neighbor of exactly one endpoint of $P$, then either $P' = (u, v, w)$ or $P' = (w, u, v)$ is the sought path. If all vertices in $U$ are neighbors of both $u$ and $v$, then $P' = (v, w)$ for any $w \in U$ is the sought path.
\end{proof}

\begin{definition}
\label{def:distance-to-cluster--letters}
Let $G$ be a graph with a modulator to cluster graph~$U$. Let $P$ be a shortest path in $G$ with $\ecc_G(P) \leq k$ and $U\cap P \neq \emptyset$.
We denote $L^P = P \cap U$ and $\bm{\pi}^P = (\pi_1^P, \dots, \pi_{\|L^P\|}^P)$ the permutation/order of vertices from~$L^P$ in which they appear on the path~$P$.
We denote $\mathcal{R}^P_i = \{u \in U \mid d_G(u, P) = i\}$ the set of vertices in $U$ that are at distance $i$ from $P$, for $i \in \{1, 2\}$.
\end{definition}

Let $V= V(G) \setminus U$. Since $G[V]$ is a disjoint union of cliques, and for every $i \in \{1, \dots, \|L^P\| - 1\}$, all vertices that are between $\pi^P_i$ and $\pi^P_{i+1}$ on $P$ are from~$V$, we have $d_G(\pi^P_i, \pi^P_{i+1}) \leq~3$, as otherwise $P$ would not be a shortest path.

Let $\bm{\pi} = (\pi_1, \dots, \pi_{\|\bm{\pi}\|})$ be a candidate (guess) on the value of $\bm{\pi}^P$.
Intuitively, if we had the correct values of $\bm{\pi}=\bm{\pi}^P$, we would only need to select the (at most two) vertices between each $\pi_i, \pi_{i+1}$. 

To help us refer to those pairs $\pi_i, \pi_{i+1}$ between which we still need to choose some vertices we  denote $\bm{h}_{\pi} = (h_1, \dots, h_\ell)$ the increasing sequence of all indices~$i$ such that $\{\pi_i, \pi_{i+1}\} \notin E$.
For every $i \notin \bm{h}_{\pi}$, we have $\{\pi_i, \pi_{i+1}\} \in E(G)$ and, thus, there is no vertex between $\pi_i$ and $\pi_{i+1}$ on $P$.
For every $h_i \in \bm{h}_{\pi}$:
If $d_G(\pi_{h_i}, \pi_{h_i+1}) = 2$, then there is one vertex on $P$ between $\pi_{h_i}$ and $\pi_{h_i+1}$, and it is from $V$.
If $d_G(\pi_{h_i}, \pi_{h_i+1}) = 3$, then there are two vertices from $V$ on $P$ between $\pi_{h_i}$ and $\pi_{h_i+1}$. 

\begin{definition}
\label{def:distance-to-cluster--candidate-set}
We define the set $C_{h_i}$ of candidate vertices between $\pi_{h_i}$ and $\pi_{h_i+1}$ for each $h_i \in \bm{h}_\pi$.
\[C_{h_i} = \begin{cases}
    \Big\{(u, u) \in V^2 \mid \big\{\{\pi_{h_i}, u\}, \{u, \pi_{h_i+1}\}\big\} \subseteq E\Big\}     & \text{if } d_G(\pi_{h_i}, \pi_{h_i+1}) = 2\\
    \Big\{(u, v) \in V^2 \mid \big\{\left\{\pi_{h_i}, u\right\}, \{u, v\}, \{v, \pi_{h_i+1}\}\big\} \subseteq E\Big\}       & \text{if } d_G(\pi_{h_i}, \pi_{h_i+1}) = 3\\
    \ \emptyset       & \text{otherwise}
\end{cases}\]
\end{definition}

For $h_i \in \bm{h}_\pi$ with $d_G(\pi_{h_i}, \pi_{h_i+1}) = 2$, the set $C_{h_i}$ contains pairs of the same vertices $(u, u)$. To avoid adding some vertex into a path twice, we define a function $\mu$ which maps a pair of two elements to a sequence of length 1 or 2:
\[\mu(u, v) = \begin{cases}
    (u)     & \text{if } u = v,\\
    (u, v)  & \text{if } u \neq v.
\end{cases}\]

To solve MESP, we need to choose exactly one pair from each of $C_{h_1}, \dots, C_{h_\ell}$.
Later, we show that the problem of choosing these pairs is an instance of CSC.

First, we define a function $\delta^P: U \cup V \rightarrow \mathbb{N}$ that will help us prove that the path constructed from the CSC solution will have a small eccentricity:
\[\delta^P(u) = \min \big\{d_G(u, L^P),\ d_G(u, \mathcal{R}^P_1) + 1,\ d_G(u, \mathcal{R}^P_2) + 2\big\}.\]

\begin{lemma}
\label{th:distance-to-cluster--e-eq-dst}
Function $\delta^P$ is a good estimate of the distance from $P$, meaning that:
\begin{enumerate}
    \item $\delta^P(u) = d_G(u, P)$ for every $u \in U$, and
    \item $\delta^P(u) = d_G\big(u, P \setminus (N_G[u] \cap V)\big)$ for every $u \in V$.
\end{enumerate}
\end{lemma}

\toappendix{
\begin{proof}
Clearly, $d_G(u, P) \leq d_G\big(u, P \setminus (N_G[u] \cap V)\big) \leq \delta^P(u)$.

Let $z$ be the nearest vertex to $u$ on $P$ and $Q$ be the shortest path from $u$ to $z$.
If there are any vertices from $U$ on~$Q$, let $x$ be the last vertex from $U$ on~$Q$.
We know that $d_G(x, z) \leq 2$ because $Q$ is a shortest path and all vertices connected in $G[V]$ form a clique.
If $x = z$, then $x \in L^P$. If $d_G(x, z) = 1$, then $x \in \mathcal{R}^P_1$.
If $d_G(x, z) = 2$, then $x \in \mathcal{R}^P_2$.
Hence, $\delta^P(u) \leq d_G(u, z) = d_G(u, P)$.
If~$Q$ consists only of vertices from $V$, let $s$ be the nearest vertex to $u$ such that $s \in P \setminus (N_G[u] \cap V)$.
Clearly, $s \in L^P$ and $\delta^P(u) \leq d_G(u, s) = d_G\big(u, P \setminus (N_G[u] \cap V)\big)$.
\end{proof}
}

Now we show how to choose optimal vertices from each $C_i$ by solving CSC.

\begin{lemma}
\label{th:distance-to-cluster--path-from-csc}
Suppose that $P$ is a shortest path in $G$ with $\ecc_G(P) \leq k$, both endpoints of $P$ are in $U$, and we have the corresponding values of $L^P, \pi^P, \mathcal{R}^P_1, \mathcal{R}^P_2$ as described in \autoref{def:distance-to-cluster--letters}. Let $\bm{h}_{\pi^P} = (h_1, \dots, h_\ell)$ and $(s_{h_1}, \dots, s_{h_\ell})$ be a solution of the CSC instance with requirements $\mathcal{R}^P = \mathcal{R}^P_1 \cup \mathcal{R}^P_2$, sets of candidates $\mathcal{C} = C_{h_1} \cup \dots \cup C_{h_\ell}$, and function $\Psi(u, v) = N_G(u) \cup N_G(v) \cup \Big(\big(N_G^2[u] \cup N_G^2[v]\big) \cap \mathcal{R}^P_2\Big)$. 
Then 
\begin{align*}
    P' =    & \ (\pi^P_1, \dots, \pi^P_{h_1}) \frown \mu(s_{h_1}) \frown (\pi^P_{h_1+1}, \dots, \pi^P_{h_2}) \\
    \dots   
    \frown  & \ \mu(s_{h_i}) \frown (\pi^P_{h_i+1}, \dots, \pi^P_{h_{i+1}}) \frown \mu(s_{h_{i+1}}) \\
    \dots   
    \frown  & \ (\pi^P_{h_{\ell-1}+1}, \dots, \pi^P_{h_\ell}) \frown \mu(s_{h_\ell}) \frown (\pi^P_{h_\ell+1}, \dots, \pi^P_{\|L^P\|})
\end{align*}
is a shortest path and $\ecc_G(P') \leq \max \{2, k\}$.
\end{lemma}

\toappendix{
\begin{proof}
Clearly, $P'$ is a shortest path.

Thanks to the way we chose $s_{h_1}, \dots s_{h_\ell}$ and from \autoref{th:distance-to-cluster--e-eq-dst} we know that:
\begin{enumerate}
    \item for every $u \in U: d_G(u, P') \leq \delta^P(u) = d_G(u, P)$,
    \item for every $u \in V: d_G(u, P') \leq \delta^P(u) = d_G\big(u, P \setminus (N_G[u] \cap V)\big)$.
\end{enumerate}
If $u \in V$ and $P \cap (N_G[u] \cap V) \neq \emptyset$, then $d_G(u, P) \leq 1$. Because $P$ contains at least one vertex from $U$ and all vertices in $N_G[u] \cap V$ form a clique, we have $d_G\big(u, P \setminus (N_G[u] \cap V)\big) \leq 2$.
\end{proof}
}

Clearly, if $k \geq 2$, then we can use \autoref{th:distance-to-cluster--path-from-csc} to construct a shortest path with eccentricity at most $k$. Now, we discuss the case when $k = 1$.

\begin{observation}
\label{th:distance-to-cluster--ecc-1-e-2}
If $\ecc_G(P) = 1$, then for every $u \in U \cup V$ we have $\delta^P(u) \leq~2$.
\end{observation}

\toappendix{
\begin{proof}
If $u \in U$, then either $u \in L^P$ or $u \in \mathcal{R}^P_1$, and $\delta^P(u) \leq~1$.
For each $u \in V$, let $v \in L^P$ be the nearest vertex to $u$ on $P$ from~$U$. Because all vertices that are connected in $G[V]$ form a clique, we get $d_G(u, v) \leq 2$, and hence $\delta^P(u) \leq 2$.
\end{proof}
}

\begin{corollary}
\label{th:distance-to-cluster--remove-candidates-when-k-1}
If $\ecc_G(P) = 1$, then a path $P'$ with $\ecc_G(P') \leq 1$ can be constructed similarly as in \autoref{th:distance-to-cluster--path-from-csc} with the following modification. For each candidate set $C_i$ which contains some pair $(x, y) \in V^2$ such that there is a neighbor
$z \in V$ of $x$ with $\delta^P(z) = 2$, remove every $(u, v) \in V^2$ such that $z$ is \emph{not} a neighbor of $u$
from $C_i$.
\end{corollary}

\toappendix{
\begin{proof}
For any $C_i$, if any of the removed pairs were selected into $P'$, then the distance of $z$ to $P'$ would be $d_G(z, P') = 2$ and therefore $\ecc_G(P') > 1$.
\end{proof}
}

We have shown how to construct a shortest path with eccentricity at most~$k$ by solving the CSC problem, even if $k = 1$. Finally, we observe that such a path can be constructed even if one or both of its endpoints are in $V$.

\begin{lemma}
\label{th:distance-to-cluster--path-from-csc--one-in-v}
If $P$ has an endpoint $s \in V$, its neighbor $t \in P$ might also be in~$V$.
Let $P = (s, t, \dots)$.
We may obtain a path~$P'$ with $\ecc_G(P') \leq k$ by removing $\Psi(s,s)$ (and $\Psi(t,t)$ if $t \in V$) from $\mathcal{R}^P$, finding $s_{h_1}, \dots, s_{h_\ell}$ by solving the CSC, and prepending $s$ (and $t$ if $t \in V$) to $P'$.
\end{lemma}

\toappendix{
\begin{proof}
All vertices in $L^P$ are on $P'$, and, thanks to the way we chose $s_{h_1}, \dots, s_{h_\ell}$, all vertices in $\mathcal{R}^P_i$ are at distance $i$ from either $s$, $t$, or one of $s_{h_1}, \dots, s_{h_\ell}$. Thus, the same argument as in the proof of \autoref{th:distance-to-cluster--path-from-csc} applies.
\end{proof}
}

MESP can be solved by trying all possible combinations of $(L, s, \mathcal{R}): L \subseteq U, s \in L, \mathcal{R} = \mathcal{R}_1 \cup \mathcal{R}_2 \subseteq (U \setminus L)$.
For each combination, do:
\begin{enumerate}
    \item Find a permutation $\bm\pi$ of $L$, such that $\pi_1 = s$ and $\sum_{i=1}^{\|L\| - 1}d_G(\pi_i, \pi_{i+1}) = d_G(\pi_1, \pi_{\|L\|})$.
    If it does not exist, continue with the next combination.
    \item For each $h_i \in \bm{h}_{\pi}$: create set $C_{h_i}$ according to \autoref{def:distance-to-cluster--candidate-set} and \autoref{th:distance-to-cluster--remove-candidates-when-k-1}.
    \item Solve the CSC instance as described in \autoref{th:distance-to-cluster--path-from-csc}.
    \item If the CSC instance has a solution, construct path $P'$ as in \autoref{th:distance-to-cluster--path-from-csc}.
    \item Check if $\ecc_G(P') \leq k$. If yes, return $P'$. If not, try the same after prepending and/or appending all combinations of single vertices and of pairs of vertices to $P'$ (see \autoref{th:distance-to-cluster--path-from-csc--one-in-v}).
\end{enumerate}

Note that by \autoref{th:unique-permutation}, there is at most one such permutation $\bm\pi$ in step 1.

\begin{theorem}
\label{th:distance-to-cluster--complexity}
In a graph with distance to cluster graph $p$, MESP can be solved in $\mathcal{O}(2^{4p}p \cdot n^6)$ time.
\end{theorem}

\toappendix{
\begin{proof}
First, we precompute a distance matrix in $\mathcal{O}(n^3)$ time.
There are at most $\mathcal{O}(4^p \cdot n)$ different combinations of $(L, s, \mathcal{R}_1, \mathcal{R}_2)$.
By \autoref{th:unique-permutation}, step 1 can be implemented in $\mathcal{O}(p \log p)$ time.
In step 2, there are at most $\mathcal{O}(p)$ sets $C_{h_i}$ and each of these can be constructed in $\mathcal{O}(n+m)$ time.
In step 3, the CSC instance can be solved by \autoref{alg:constrained-set-cover} in $\mathcal{O}(2^{2p}pn^2)$ time as $\|\mathcal{R}\| \leq p$ and $\|\mathcal{C}\| \leq n^2$.
Step 4 can be implemented in $\mathcal{O}(n)$ time.
In step 5, there are at most $\mathcal{O}(n^4)$ combinations of vertices to prepend/append, and checking the length of the resulting path and its eccentricity can be implemented in $\mathcal{O}(n)$ time using the precomputed distance matrix.
\end{proof}
}

\subsection{Distance to Disjoint Paths}
\label{sec:distance-to-disjoint-paths}
\toappendix{
}
In this subsection, we present an fpt-algorithm for MESP parameterized by the distance to disjoint paths and the desired eccentricity, combined. While the existence of such an fpt-algorithm is implied by the results of \autoref{sec:treewidth}, the treewidth algorithm uses Courcelle's theorem and, as such, is rather of classification nature. In comparison, in this section, we present an explicit algorithm with moderate dependency on the parameters.

The high-level idea of the algorithm is similar to that in \autoref{sec:distance-to-cluster}.
We iteratively guess (by trying all possible combinations), for each vertex in the modulator to disjoint paths $C$, what is the distance to the desired shortest path.
Then, we try to find a shortest path which satisfies all the guessed distance requirements by solving an instance of the CSC problem.
We argue that if these requirements are guessed correctly, the resulting path will indeed be the desired MESP.

We start by discussing some properties of graphs in which a shortest path~$P = (p_1, \dots, p_{\|P\|})$ with $\ecc_G(P) \leq k$ does exist. Assume that $P$ is such a path, fixed for the next few lemmas and definitions.

\begin{definition}
Let $\widehat{C}^P = C \cup \{p_1, p_{\|P\|}\}$.
Let $L^P = P \cap \widehat{C}^P$.
We denote $\bm\pi^P = (\pi^P_1, \dots, \pi^P_{\|L^P\|})$ the permutation/order of vertices from $L^P$ on the path $P$.
We define function $\delta^P(v) = d_G(v, P)$ for every $v \in V$.
\end{definition}

Let $\widehat{C}$, $L$ be candidates for $\widehat{C}^P$, $L^P$, respectively.
Similarly as in \autoref{sec:distance-to-cluster}, the permutation $\bm\pi = (\pi_1, \dots, \pi_{\|L\|})$ of the vertices in $L$ is unique (if it exists), and can be found in polynomial time.
For each consecutive pair of vertices $\pi_i, \pi_{i+1} \in L$, there may be multiple shortest paths connecting them, such that they do not contain any other vertices from $\widehat{C}$. Exactly one of these shortest paths is contained in $P$ for each pair.
We say $\bar\sigma$ is a \emph{candidate segment} if it is a sequence of vertices on some shortest path from $\pi_i$ to $\pi_{i+1}$ excluding the endpoints $\pi_i, \pi_{i+1}$ and $\bar\sigma \cap \widehat{C} = \emptyset$. We define ${\mathcal{S}}(\pi_i, \pi_{i+1})$ as a set of all candidate segments $\bar\sigma$ between $\pi_i$ and $\pi_{i+1}$. We denote $\widetilde{\mathcal{S}} = \bigcup_{i=1}^{\|L\| - 1}{\mathcal{S}}(\pi_i, \pi_{i+1})$ the set of all candidate segments in $G$. 
We say that a candidate segment $\sigma \in \widetilde{\mathcal{S}}$ is a \emph{necessary segment} if
it must be part of any shortest path $P'$ such that $\widehat{C}=\widehat{C}^{P'}$, $L = L^{P'}$, $\bm\pi = \bm{\pi}^{P'}$, and $\ecc_G(P') \le k$.

Intuitively, if we had the correct values of $\bm\pi$, we would only need to select one segment out of each ${\mathcal{S}}(\pi_i, \pi_{i+1})$ for $i \in \{1, \dots, \|L\| - 1\}$, in order to construct the path $P$.
To do so, we need the following function, which estimates the distance from a vertex to the path $P$.

\begin{definition}[estimate distance to $P$]
For a graph $G = (V, E)$, a set of vertices $\widehat{C} \subseteq V$ and a function $\delta: \widehat{C} \rightarrow \mathbb{N}$ we define $d_G^\delta: V \times 2^{\widehat{C}} \rightarrow \mathbb{N}$ as
\[d_G^\delta(v, S) = \min_{s \in S}d_G(v, s) + \delta(s).\]
\end{definition}

\begin{observation}
\label{th:distance-to-disjoint-paths--d-le-de}
If $\widehat{C}= \widehat{C}^P$ and $\delta=\delta^P\|_{\widehat{C}}$ (that is, the restriction of $\delta^P$ to $\widehat{C}$) for some shortest path $P$ in $G$, then for every $v \in V$ we have $d_G(v, P) \leq d_G^\delta(v, \widehat{C})$.
In particular, if $d_G^\delta(v, \widehat{C}) \leq k$, then $d_G(v, P) \leq k$.
\end{observation}

\toappendix{
\begin{proof}
By definition, for any $v \in V$, there is some $s \in \widehat{C}$ such that $d_G^\delta(v, \widehat{C}) = d_G(v, s) + \delta(s) = d_G(v, s) + d_G(s, P)$ and, from triangle inequality, $d_G(v, P) \leq d_G(v, s) + d_G(s, P)$.
\end{proof}
}

If we had the correct values for the permutation $\bm\pi$ of vertices from $\widehat{C}$ that are on $P$, we would still have to take care of those vertices $v \in V$ with $d_G^e(v, \widehat{C}) > k$, in order to solve  MESP. In particular, we would have to choose a segment from each ${\mathcal{S}}(\pi_i, \pi_{i+1})$ in a way that for every vertex~$v$ with $d_G^e(v, \widehat{C}) > k$, there would be some chosen segment at distance at most~$k$ from~$v$.
We say that a candidate segment $\bar\sigma \in \widetilde{\mathcal{S}}$ \textit{satisfies} $v \in V \setminus \widehat{C}$ if $d_G(v, \bar\sigma) \leq k < d_G^\delta(v, \widehat{C})$.

We continue by showing that the number of vertices $v$ with $d_G^\delta(v, \widehat{C}) > k$ which \emph{do not} lie on $P$ is bounded by the size of $L$.

\begin{lemma}
\label{th:disjoint-paths--sat-size-outside-P}
Let $\bar\sigma \in \widetilde{\mathcal{S}}$ be a candidate segment and $D = \{v \in V \setminus P \mid \bar\sigma$ satisfies $v\}$. Then $\|D\| \leq 2$.
\end{lemma}

\toappendix{
\begin{proof}
Let $v \in D$ and $u \in \bar\sigma$ be the nearest vertex to $v$ on segment~$\bar\sigma$. There is a shortest path from $u$ to $v$ which does not contain any vertex from $\widehat{C}$ (if it did, then $d_G^\delta(v, \widehat{C}) = d_G(v, P) \leq k$). Because $G \setminus C$ is a union of disjoint paths and it contains the whole segment $\bar\sigma$ as well as the path from $u$ to $v$, these two paths must be connected through their endpoints. Neither of the endpoints is in $C$, so no more than two such connections can be present in~$G$ (see \autoref{fig:disjoint-paths--sat-max-2}).

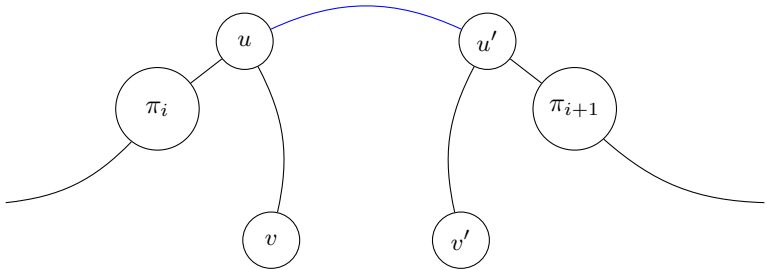
\begin{figure}[t]
\centering
\begin{tikzpicture}
  [auto=left,every node/.style={circle,draw=black,minimum size=0.75cm}]
  \node [style={minimum size=1.1cm}] (pi_i) at (1,8.75) {$\pi_i$};
  \node [style={minimum size=1.1cm}] (pi_ii) at (6.5,8.75) {$\pi_{i+1}$};
  \node (u) at (2.15,9.65) {\small$u$};
  \node (u_) at (5.35,9.65) {\small$u'$};
  \node (v) at (2.5,7)  {\small$v$};
  \node (v_) at (5,7)  {\small$v'$};
  \draw (-1, 7.5) edge[bend right=20] (pi_i);
  \draw (pi_ii) edge[bend right=20] (9,7.5);
  \draw (pi_i) edge (u);
  \draw (u) edge[blue, bend left=25] (u_);
  \draw (u_) edge (pi_ii);
  \draw (v) edge[bend right=20] (u);
  \draw (v_) edge[bend left=20] (u_);
\end{tikzpicture}
\caption[Example of a situation from \autoref{th:disjoint-paths--sat-size-outside-P}]{Example of a situation from \autoref{th:disjoint-paths--sat-size-outside-P}. Segment $\bar\sigma$ is highlighted in blue and $D = \{v, v'\}$.}
\label{fig:disjoint-paths--sat-max-2}
\end{figure}
\end{proof}
}

\begin{corollary}
\label{th:disjoint-paths--u-size}
Let $P$ be a shortest path in $G$ with $ecc_G(P) \leq k$.
Let $U = \{v \in V \setminus P \mid d_G^{\delta^P}(v, \widehat{C}^P) > k\}$. There are $\|L^P\| - 1$ segments on $P$, therefore $\|U\| \leq 2(\|L^P\| - 1)$.
\end{corollary}

We have shown that there are not many vertices $v \notin P$ with $d_G^{\delta^P}(v, \widehat{C}^P) > k$.
Now, we show that all such vertices actually have $d_G^{\delta^P}(v, \widehat{C}^P) = k + 1$.

\begin{lemma}
\label{th:disjoint-paths--distance-to-mesp}
Let $P$ be a shortest path in $G$ with $ecc_G(P) \leq k$.
Let $v \in V$ be such that $d_G^{\delta^P}(v, \widehat{C}^P) \geq k + 1$. Let $u \in P$ be the nearest vertex to $v$ on $P$. Then either $u = v$, or $d_G(u, v) = k$ and $d_G(u, \widehat{C}^P) = 1$.
\end{lemma}

\toappendix{
\begin{proof}
If $v \in P$, then the nearest vertex on $P$ is itself, so $u = v$.
Suppose that $v \notin P$ (see \autoref{fig:disjoint-paths--edge-case}). There is no vertex from $\widehat{C}^P$ on any shortest path between $v$~and~$u$ (otherwise $d_G^{\delta^P}(v, \widehat{C}^P) \leq d_G(v, u) \leq k$). In particular, $u \notin \widehat{C}^P$, therefore, $u$ has exactly 2 neighbors on $P$. It also has at least one neighbor outside of $P$, through which it is connected to $v$. In $G \setminus \widehat{C}^P$, $u$ must have at most 2 neighbors, thus at least one of its neighbors on $P$ is in $\widehat{C}^P$. Because $L^P = P \cap \widehat{C}^P$, we get $d_G(u, L^P) = 1$. Then, from
\[k + 1 \leq d_G^{\delta^P}(v, \widehat{C}^P) \leq d_G(v, u) + d_G(u, L^P) = d_G(v, u) + 1 \leq k + 1,\]
we get $d_G(v, u) = k$.
\end{proof}
}

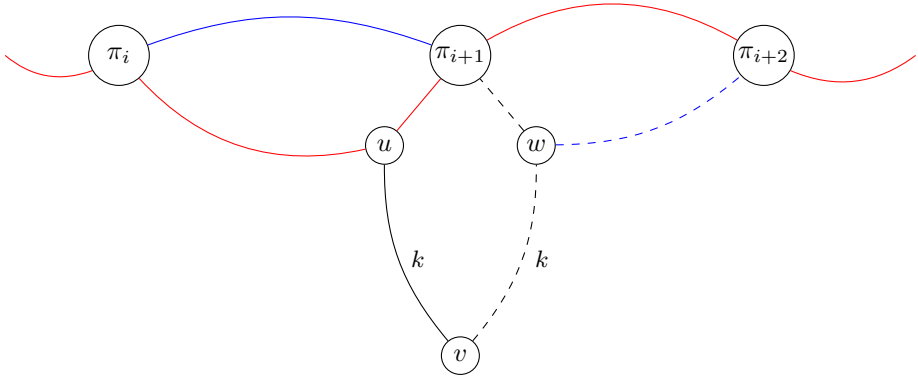
\begin{figure}[t]
\centering
\begin{tikzpicture}
  [auto=left,every node/.style={inner sep=0, minimum size=.5cm, circle,draw=black}]
  \node [style={minimum size=.8cm}] (pi_i) at (1.5,10) {$\pi_i$};
  \node [style={minimum size=.8cm}] (pi_ii) at (6,10) {$\pi_{i+1}$};
  \node [style={minimum size=.8cm}] (pi_iii) at (10, 10) {$\pi_{i+2}$};
  \node (u) at (5,8.8) {$u$};
  \node (v) at (6,6)  {$v$};
  \node (w) at (7, 8.8) {$w$};
  \draw (0, 10) edge[red, bend right=30] (pi_i);
  \draw (pi_i) edge[blue, bend left=20] (pi_ii);
  \draw (pi_i) edge[red, bend right=30] (u);
  \draw (u) edge[red] (pi_ii);
  \draw (pi_ii) edge[red, bend left=30] (pi_iii);
  \draw (pi_iii) edge[red, bend right=30] (12, 10);
  \draw (v) edge[bend left=20] node[right, draw=none] {\small$k$} (u);
  \draw (v) edge[dashed, bend right=20] node[right, draw=none] {\small$k$} (w);
  \draw (w) edge[dashed] (pi_ii);
  \draw (w) edge[blue, dashed, bend right=20] (pi_iii);
\end{tikzpicture}
\caption[Example of a situation from \autoref{th:disjoint-paths--distance-to-mesp}]{Example of a situation from \autoref{th:disjoint-paths--distance-to-mesp}. Path $P$ is red, candidate segments are blue. The segment containing vertex $u$ satisfies $v$. If all the dashed parts are present in $G$, then the segment containing vertex $w$ also satisfies $v$.}
\label{fig:disjoint-paths--edge-case}
\end{figure}

Let $v \in V \setminus P$ be such that $d_G^{\delta^P}(v, \widehat{C}^P) = k + 1$ and $S \subseteq \widetilde{\mathcal{S}}$ be a set of candidate segments that satisfy $v$. As shown in \autoref{fig:disjoint-paths--edge-case}, there may be multiple such segments in $S$.
However, if $v$ itself lies on some segment, then only segments in the same ${\mathcal{S}}(\pi_i, \pi_{i+1})$ may satisfy $v$.

\begin{observation}
\label{th:distance-to-disjoint-paths--est-plus-1-sat-same-pi-i}
Let $P$ be a shortest path in $G$ with $ecc_G(P) \leq k$.
Let $\bar\sigma \in {\mathcal{S}}(\pi^P_i, \pi^P_{i+1})$ be a candidate segment that contains some vertex $u$ such that $d_G^{\delta^P}(u, \widehat{C}^P) = k + 1$. Let $u' \in P$ be the nearest vertex to $u$ on $P$. Then $u'$ lies on a segment $\sigma \in {\mathcal{S}}(\pi^P_i, \pi^P_{i+1})$.
\end{observation}

\toappendix{
\begin{proof}
If $u' \in L$, then $d_G^\delta(u, \widehat{C}) \leq d_G(u, u') \leq k$.
If $\sigma \notin {\mathcal{S}}(\pi^P_i, \pi^P_{i+1})$, then~$P$ would not be a shortest path because $d_G(u, u') \leq k < k + 1 \leq d_G(u, \{\pi^P_i, \pi^P_{i+1}\})$.
\end{proof}
}

We already know from \autoref{th:disjoint-paths--distance-to-mesp} that if a candidate segment contains some vertex $v$ with $d_G^\delta(v, \widehat{C}) > k + 1$, then it is a necessary segment. Now, we show another sufficient condition for a candidate segment to be a necessary segment.

\begin{lemma}
\label{th:distance-to-disjoint-paths--two-est-plus-one-implies-true-segment}
Let $\bar\sigma \in {\mathcal{S}}(\pi_i, \pi_{i+1})$ be a candidate segment that contains some vertices $u, v \in \bar\sigma$ such that $u \neq v$ and $d_G^\delta(u, \widehat{C}) = d_G^\delta(v, \widehat{C}) = k + 1$. Then, $\bar\sigma$~is a necessary segment.
\end{lemma}

\toappendix{
\begin{proof}
Suppose that $\bar\sigma$ is not a necessary segment, let $P$ be a shortest path in $G$ with $ecc_G(P) \leq k$, $\widehat{C}=\widehat{C}^{P}$, $L = L^{P}$, $\bm\pi = \bm{\pi}^{P}$, and $\delta^P\|_{\widehat{C}^P}=\delta$ and let $\sigma \in {\mathcal{S}}(\pi_i, \pi_{i+1})$ be a part of $P$, $\sigma \neq \bar\sigma$.
By \autoref{th:disjoint-paths--distance-to-mesp} there must be some vertices $u', v' \in P$ with $d_G(u, u') = d_G(v, v') = k$.
By \autoref{th:distance-to-disjoint-paths--est-plus-1-sat-same-pi-i}, both $u'$ and $v'$ are in~$\sigma$. 
No shortest path between $u$ and $u'$, contains any vertex from $\widehat{C}$ (otherwise $d_G^\delta(u, \widehat{C}) \leq d_G(u, u') =  k$). The same applies for any shortest path between $v$ and $v'$.
Thus, in $G \setminus \widehat{C}$, $u$ is connected to~$v$, $v$~is connected to~$v'$, $v'$~is connected to~$u'$, and~$u'$~is connected to~$u$. Due to the length constraints, all these paths are disjoint (except for their endpoints). There is a cycle in $G \setminus \widehat{C}$, which is a contradiction with $C$ being a modulator to disjoint paths.
\end{proof}
}

Let us summarize what we have shown so far.
If we had the correct values for the permutation $\bm\pi$ of vertices from $\widehat{C}$ that are on $P$, and of $\delta$, we would only need to select one segment out of each ${\mathcal{S}}(\pi_i, \pi_{i+1})$ to find a shortest path with eccentricity at most $k$.
There are some vertices $u \in V$ such that $d_G^\delta(u, \widehat{C}) \leq k$ and for these vertices, the distance to the resulting path will be at most $k$, no matter which segments we choose.

A segment which contains some vertex $v$ with $d_G^\delta(v, \widehat{C}) > k + 1$ is a necessary segment. A segment which contains two vertices $u \neq v$ with $d_G^\delta(u, \widehat{C}) = d_G^\delta(v, \widehat{C}) = k + 1$ is a necessary segment as well.
For the remaining segments, we know that for every $u \in V$ with $d_G^\delta(u, \widehat{C}) > k$, the shortest path with eccentricity at most $k$ needs to contain some $\sigma_u \in \widetilde{\mathcal{S}}$ such that $d_G(u, \sigma_u) \leq k$. Furthermore, for every $v \in \widehat{C}$ with $d_G(v, L) > \delta(v)$, the path needs to contain some $\sigma_v \in \widetilde{\mathcal{S}}$ such that $d_G(v, \sigma_v) \leq \delta(v)$.

Clearly, the problem of selecting one segment out of each set of candidate segments is an instance of CSC: the sets of candidates are $\mathcal{C} = {\mathcal{S}}(\pi_1, \pi_2) \cup \dots \cup {\mathcal{S}}(\pi_{\|L\|-1}, \pi_{\|L\|})$, the requirements are $\mathcal{R} = \{v \in V \setminus \widehat{C} \mid d_G^\delta(v, \widehat{C}) > k\} \cup \{v \in \widehat{C} \setminus L \mid d_G(v, L) > \delta(v)\}$, and the function $\Psi(\bar\sigma) = \{v \in V \setminus \widehat{C} \mid \bar\sigma \text{ satisfies } v\} \cup \{v \in \widehat{C} \setminus L \mid d_G(v, \sigma) \leq \delta(v)\}$.

We know that the number of vertices outside of $P$ that the segments can satisfy is bounded by the size of $L$. Furthermore, we know that if a segment contains at least two vertices that need to be satisfied, then it is a necessary segment. Lastly, we know that if a segment from some ${\mathcal{S}}(\pi_i, \pi_{i+1})$ contains one vertex $v$ with $d_G^\delta(v, \widehat{C}) = k + 1$, then only segments from the same ${\mathcal{S}}(\pi_i, \pi_{i+1})$ may satisfy $v$. Thus, all segments in ${\mathcal{S}}(\pi_i, \pi_{i+1})$ that do not satisfy $v$ may be disregarded. By this, we ensure that $v$ will be satisfied no matter which segment is chosen, and $v$ does not need to be added to the requirements~$\mathcal{R}$. Hence, the requirements $\mathcal{R}$ do not need to contain any vertices from $P$, and the size of $\mathcal{R}$ is bounded by the size of $C$.

In the following lemma, we show that we do not need to explicitly check whether a segment contains some vertex $v$ with $d_G^\delta(v, \widehat{C}) > k + 1$ to decide that it is a necessary segment. This will simplify our algorithm a bit.

\begin{lemma}
\label{th:distance-to-disjoint-paths--est-big-implies-plus-one}
If a segment $\sigma \in \widetilde{\mathcal{S}}$ contains a vertex $u$ such that $d_G^\delta(u, \widehat{C}) > k + 1$, then it must also contain two vertices $v, v'$ with $d_G^\delta(v, \widehat{C}) = d_G^\delta(v', \widehat{C}) = k + 1$.
\end{lemma}

\toappendix{
\begin{proof}
Let $s, t$ be endpoints of $\sigma$, thus $d_G^\delta(s, \widehat{C}) = d_G^\delta(t, \widehat{C}) = 1$. If there was no $v \in \sigma$ with $d_G^\delta(v, \widehat{C}) = k + 1$ between $s$ and $u$, then there would have to be some neighbors $p, q \in \sigma$ such that $d_G^\delta(p, \widehat{C}) > d_G^\delta(q, \widehat{C}) + 1$. This is a contradiction because clearly $d_G^\delta(p, \widehat{C}) \leq d_G^\delta(q, \widehat{C}) + 1$ if $p$ is a neighbor of~$q$. The same holds for $v' \in \sigma$ with $d_G^\delta(v', \widehat{C}) = k + 1$ between $u$ and $t$.
\end{proof}
}

Finally, we propose an algorithm that solves MESP. 
It finds the correct values for $p_1, \widehat{C}, L$, and $\delta \colon \widehat{C} \to \{0, \ldots, k\}$ by trying all possible combinations. For each combination, it performs the following steps.
\begin{enumerate}
    \item Find a permutation $\bm\pi$ of $L$, such that $\pi_1 = p_1$ and $\sum_{i=1}^{\|L\| - 1}d_G(\pi_i, \pi_{i+1}) = d_G(\pi_1, \pi_{\|L\|})$.
    If it does not exist, continue with the next combination.
    \item For each $\pi_i, \pi_{i+1}$, check all candidate segments in ${\mathcal{S}}(\pi_i, \pi_{i+1}).$
    \begin{enumerate}
        \item If there are any segments containing a vertex $u$ with $d_G^\delta(u, \widehat{C}) = k + 1$, then we may disregard all candidate segments which do not satisfy~$u$.
        \item After disregarding these segments, if there is only one candidate segment left, it is a necessary segment. If there is no candidate segment left, then no solution exists.
    \end{enumerate}
    \item If there is a vertex $v$ such that $d_G^\delta(v, \widehat{C}) > k + 1$, and it does not lie on a segment that we have marked as a necessary segment, then no solution exists.
    \item Construct the set $U$ of vertices $v$ that are not contained in any segment and have $d_G^\delta(v, \widehat{C}) = k + 1$. If $\|U\| > 2(\|L\| - 1)$, then no solution exists.
    \item Choose the rest of the segments from all candidate segments (except those disregarded in step 1) by solving the CSC instance, with requirements $u \in \widehat{C} \setminus L$ whose distance to the parts of $P$ selected so far is greater than $\delta(u)$, and all of~$U$.
    \item If the CSC instance has a solution, construct a path from $\bm{\pi}$ and from the chosen candidate segments. If the resulting path has eccentricity at most $k$, return it.
\end{enumerate}

Note that by \autoref{th:unique-permutation}, there is at most one such permutation $\bm\pi$ in step 1.

We arrive at the following theorem.

\begin{theorem}
\label{th:distance-to-disjoint-paths--complexity}
For a graph with distance to disjoint paths $c$, MESP can be solved in $\mathcal{O}(2^{5c}k^cc \cdot n^4)$.
\end{theorem}

\toappendix{
\begin{proof}
The distances between all pairs of vertices can be precomputed in $\mathcal{O}(n^3)$ time.

There are at most $\mathcal{O}(n^2)$ possible combinations for the first and last vertex on the path giving $\widehat{C}$.
For each of them, there at most $\mathcal{O}(2^ck^c)$ possible values of $L$ and $\delta$.

By \autoref{th:unique-permutation}, step 1 can be implemented in $\mathcal{O}(c \log c)$ time.

In step 2, we iterate over $i \in \{1, \dots, \|\bar{L} - 1\|\}$, and each time we find a set $K$ of all such vertices $u$ that lie on some segment $\bar\sigma \in {\mathcal{S}}(\bar\pi_i, \bar\pi_{i+1})$ and have $d_G^\delta(u, \widehat{C}) = k + 1$. The set can be constructed in $\mathcal{O}(n + m)$ time.
Because of \autoref{th:distance-to-disjoint-paths--two-est-plus-one-implies-true-segment}, at most two vertices need to be stored in $K$ for each candidate segment. Thus, by \autoref{th:disjoint-paths--sat-size-outside-P}, each segment will satisfy at most four vertices from $K$ (two lying on the segment, and two outside the segment). If $\|K\| > 4$, we know right away that there is no solution for the current configuration of $(L, \pi, \delta)$. Otherwise, we construct the set $C_i$ of candidate segments which satisfy all vertices from $K$ in $\mathcal{O}(n + m)$ time by iterating over all vertices in all the candidate segments and checking at most 4  conditions for each vertex.
Thus, the complexity of step 2 is $\mathcal{O}(cn^2)$.

Steps 3 and 4 can be performed in one loop, taking $\mathcal{O}(cn)$ time.

The CSC instance in step 5 can be solved by \autoref{alg:constrained-set-cover} in $\mathcal{O}(2^{4c}cn)$ time as $\|\mathcal{C}\| = \mathcal{O}(n)$ and
\begin{align*}
\|\mathcal{R}\| \leq   & \ \|(\widehat{C} \setminus \bar{L}) \cup \bar{U}\| \\
                = & \ \|\widehat{C}\| - \|\bar{L}\| + \|\bar{U}\| \\
                \leq & \ \|\widehat{C}\| - \|\bar{L}\| + 2(\|\bar{L}\| - 1) \\
                = & \ \|\widehat{C}\| + \|\bar{L}\| - 2 \\
                \leq & \ (c + 2) + (c + 2) - 2 \\
                = & \ 2c + 2.
\end{align*}

In step 6, constructing the path and checking its eccentricity takes at most $\mathcal{O}(n)$ time.
\end{proof}
}

\subsection{Treewidth}
\label{sec:treewidth}
In this section, we show that the MESP problem is FPT with respect to the combination of treewidth and the desired eccentricity.
We present a proof which uses Courcelle's theorem.
Although this result has theoretical significance and the proof is constructive, the multiplicative constants of the resulting algorithm are too high for it to be useful in practice, even for graphs with treewidth 1.

To use Courcelle's theorem, properties of graphs are usually described using the fragment of logic called \emph{monadic second-order logic} (MSOL). The graph is described using a universe consisting of vertices and edges, unary predicates distinguishing between these two and two binary predicates $\operatorname{adj}$ and $\operatorname{inc}$, describing the adjacencies between vertices and incidencies between vertices and edges, respectively. MSOL allows quantification over individual vertices and edges and over sets of vertices and edges, membership queries, and standard logical operators.

We use the following extension of Courcelle's theorem.
    
\begin{theorem}[Courcelle \cite{Courcelle90}; Arnborg et al.~\cite{ArnborgLS91}]
Let $G = (V, E)$ be a graph, $(T, \beta)$ a tree decomposition of width $t$, and $\varphi$ an MSOL formula over the graph language.
There is an fpt-algorithm that decides for any given $G, T, \beta, \varphi$ whether $G \models \varphi$ in time $\mathcal{O}\left(f(t, \|\varphi\|\right) \cdot n)$ for some computable function~$f$.
Moreover, if $\varphi$ contains a free variable $S$, it is possible to find a set $S'$ such that $(G, S') \models \varphi(S)$, and $\|S'\|$ is minimal, in the same running time.
\end{theorem}

\begin{lemma}
    \label{th:treewidth-msol-exists}
    For any given graph $G = (V, E)$, desired eccentricity $k \in \mathbb{N}$ and two vertices $s, t \in V$, there is an MSOL formula $\varphi(S)$ of length $\mathcal{O}(k)$ such that
    a shortest $s$-$t$ path $P$ in $G$ with $\ecc_G(P) \le k$ exists if and only if there is an assignment $S' \subseteq V$ to $S$ such that $G \models \varphi(S')$ and $\|S'\| = d_G(s, t) + 1$.
    Moreover, $S' = V(P)$.
\end{lemma}

\begin{proof}
    We need to find an MSOL formula with a free variable $S \subseteq V$ such that it is satisfied if and only if there exists an $s$-$t$ path on the vertices of $S$ with eccentricity at most $k$.
    The minimality of $\|S\|$ will then imply our lemma.
    Consider the conjunction of the following conditions:
    \begin{enumerate}[label=(\arabic*)]
        \item $\{s, t\} \subseteq S \subseteq V$, and
        \item $\forall u \in (S \setminus \{s, t\}): \deg_S(u) = 2$, and
        \item $\forall u \in \{s, t\}: \deg_S(u) = 1$, and
        \item $\ecc_G(S) \le k$.
    \end{enumerate}
    The above conditions describe a set $S$ which contains only vertices of degree 2 and two leaves ($s$ and $t$).
    In other words, $S$ contains an $s$-$t$ path and potentially a union of disjoint cycles.
    Minimizing $\|S\|$ clearly ensures that it will not contain any cycles and the $s$-$t$ path will be a shortest path, if it is still possible to satisfy condition (4).

    It remains to show that all the four conditions can be expressed in MSOL and their length is linear in $k$.
    We analyze the conditions one by one.
    \begin{enumerate}
        \item Condition (1) can be written as an MSOL formula of constant length: 
        \[s \in S \land t \in S \land \forall u \in~S: u \in V.\]
        \item Condition (2) can be written as: 
        \[\forall u \in S: u = s \lor u = t \lor (\deg_S(u) \ge 2 \land \deg_S(u) \le 2),\]
        where
        \begin{itemize}
            \item $\deg_S(u) \ge 2 \equiv \exists v, w \in V: v \ne w \land \operatorname{adj}(u, v) \land \operatorname{adj}(u, w)$ and
            \item $\deg_S(u) \le 2 \equiv \forall v, w, x \in V: v = w \lor v = x \lor w = x \lor \neg\operatorname{adj}(u, v) \lor \neg\operatorname{adj}(u, w) \lor \neg\operatorname{adj}(u, x)$.
        \end{itemize}
        This is an MSOL formula of constant length.
        \item Similarly, condition (3) can be written as: 
        \[\forall u \in \{s, t\}: \deg_S(u) \ge 1 \land \deg_S(u) \le 1,\]
        where
        \begin{itemize}
            \item $\deg_S(u) \ge 1 \equiv \exists v \in V: \operatorname{adj}(u, v)$ and
            \item $\deg_S(u) \le 1 \equiv \forall v, w \in V: v = w \lor \neg\operatorname{adj}(u, v) \lor \neg\operatorname{adj}(u, w)$.
        \end{itemize}
        This is an MSOL formula of constant length.
        \item Condition (4) can be written as: 
        \[\forall u \in V: \exists p_1, \dots, p_k \in V: \operatorname{adj}(u, p_1) \land \operatorname{adj}(p_1, p_2) \land \dots \land \operatorname{adj}(p_{k-1}, p_k) \land (p_k \in S).\]
        This is an MSOL formula of length $\mathcal{O}(k)$.
    \end{enumerate}
\end{proof}

\begin{corollary}
    There is an algorithm that solves MESP in $\mathcal{O}(f\left(t, \mathcal{O}(k)\right) \cdot n^3)$ time where $t$ is the treewidth of the input graph and $k$ is the desired eccentricity.
\end{corollary}

\begin{proof}
Evaluate the formula from \autoref{th:treewidth-msol-exists} for all $\mathcal{O}(n^2)$ pairs $\{s, t\} \subseteq V$, each in $\mathcal{O}(f\left(t, \mathcal{O}(k)\right) \cdot n)$ time.
If some of these evaluations finds a set $S'$ such that $\|S'\| = d_G(s, t) + 1$, then we have the vertices of a shortest path with eccentricity $k$, and we can construct the path by \autoref{th:unique-permutation}.
Otherwise, no solution exists.
\end{proof}

\subsection{Maximum Leaf Number}
\label{sec:max-leaf-number}
\toappendix{
}

\begin{theorem}
\label{thm:max-leaf}
There is an algorithm that solves MESP in $\mathcal{O}(16^\ell \cdot n^3)$ time, where~$\ell$ is the maximum leaf number of the input graph.
\end{theorem}

\toappendix{

We use the following lemma to show that the problem can be solved by brute-force.

\begin{lemma}[Bouland \cite{bouland_2011}]
\label{th:max-leaf-number-subdivision}
If the maximum leaf number of a graph $G$ is equal to $\ell$, then $G$ is a subdivision of a graph $\widetilde{G} = (C, \widetilde{E})$, such that $C \subseteq V$ and $\|C\| < 4 \ell$.
\end{lemma}

Our algorithm will in fact check all the shortest paths in $G$ to find the one with the smallest eccentricity.
It remains to show an upper bound on the number of shortest paths in $G$.

\begin{lemma}
\label{th:max-leaf-number--p-from-pi}
Let $\{s, t\} \subseteq V$ and $L \subseteq C$.
Then, there is at most one shortest $s$-$t$ path $P$ in $G$ such that $P \cap C = L$.
\end{lemma}

\begin{proof}
By \autoref{th:unique-permutation} there is at most one permutation $\bm\pi$ of the vertices from $L$ which corresponds to the order on some shortest path starting in $s$.

Between each pair of vertices in $C$ there is at most one path in $G$ not containing internal vertices from $C$, since $G$ is a subdivision of $\widetilde{G}$ and $\widetilde{G}$ is a simple graph.
By the same argument, there is also at most one path from $s$ to $\pi_1$, and at most one path from $\pi_{\|L\|}$ to $t$.
\end{proof}

\begin{corollary}
There are at most $2^{4\ell} \cdot n^2$ distinct shortest paths in $G$.
\end{corollary}

All the shortest paths in $G$ may be found exactly the same way as in \autoref{sec:modular-width}.
Again, with a precomputed distance matrix, finding each shortest path will take at most $\mathcal{O}(n)$ time, as well as checking its eccentricity.
This finishes the proof of \autoref{thm:max-leaf}.
}

\section{Future Directions}
\label{chap:conclusion}
We have shown that MESP is fixed-parameter tractable with respect to several structural parameters.
This partially answers an open question of Dragan and Leitert~\cite{DraganL16} on classes where the problem is polynomial time solvable, as this is the case whenever we limit ourselves to a class where one of the studied parameters is a constant.

The natural next steps in the research of parameterized complexity of MESP would be to investigate the existence of fpt-algorithms with respect to the distance to disjoint paths alone, and with respect to treewidth alone.

\section*{Acknowledgements}
The authors would like to express their thanks to Dr. Dušan Knop for fruitful discussions concerning the problem that in particular led to the discovery of the treewidth algorithm, and to Tung Anh Vu for naming the \textsc{Constrained Set Cover} problem.

The work of Martin Kučera was supported by the Student Summer Research Program 2020 of FIT CTU in Prague, and by grant SGS20/212/OHK3/3T/18.

Ondřej Suchý acknowledges the support of the OP VVV MEYS funded project CZ.02.1.01/0.0/0.0/16\_019/0000765 ``Research Center for Informatics''.

\section*{Declaration of Competing Interest}
The authors declare that they have no known competing financial interests or personal relationships that could have appeared to influence the work reported in this paper.

\bibliography{ref}



\end{document}